\documentclass[11pt,letterpaper]{article}

\usepackage{graphicx}
\usepackage{hyperref}
\usepackage{amssymb}
\usepackage{amsmath}
\usepackage{amsthm}
\newtheorem{theorem}{Theorem}

\newtheorem{lemma}[theorem]{Lemma}

\newtheorem{corollary}[theorem]{Corollary}

\def\A{\mathcal{A}}
\def\C{\mathcal{E}}

\title{Satisfiability thresholds for $k$-CNF formula with bounded variable intersections}

\author{Karthekeyan Chandrasekaran\thanks{karthe@gatech.edu} \\ \small Georgia Institute of Technology \and Navin Goyal\thanks{navingo@microsoft.com}\\ \small Microsoft Research, India \and Bernhard Haeupler\thanks{haeupler@mit.edu} \\ \small Massachusetts Institute of Technology}

\date{}

\begin{document}
\maketitle

\begin{abstract}
We determine the thresholds for the number of variables, number of clauses, number of clause intersection pairs and the maximum clause degree of a $k$-CNF formula that guarantees satisfiability under the assumption that every two clauses share at most $\alpha$ variables. More formally, we call these formulas $\alpha$-intersecting and define, for example, a threshold $\mu_i(k,\alpha)$ for the number of clause intersection pairs $i$, such that every $\alpha$-intersecting $k$-CNF formula in which at most $\mu_i(k,\alpha)$ pairs of clauses share a variable is satisfiable and there exists an unsatisfiable $\alpha$-intersecting $k$-CNF formula with $\mu_m(k,\alpha)$ such intersections. We provide a lower bound for these thresholds based on the Lov\'asz Local Lemma and a nearly matching upper bound by constructing an unsatisfiable $k$-CNF to show that $\mu_i(k,\alpha) = \tilde{\Theta}(2^{k(2+1/\alpha)})$. Similar thresholds are determined for the number of variables ($\mu_n = \tilde{\Theta}(2^{k/\alpha})$) and the number of clauses ($\mu_m = \tilde{\Theta}(2^{k(1+\frac{1}{\alpha})})$) (see \cite{scheder08almostdisjoint} for an earlier but independent report on this threshold). Our upper bound construction gives a family of unsatisfiable formula that achieve all four thresholds simultaneously.
\end{abstract}

\bigskip

\section{Introduction}

Satisfiability of CNF is one of the most studied and versatile problems in computer science with its own journal (JSAT), competitions and an yearly conference, International Conference on Theory and Applications of Satisfiability Testing (SAT). In this paper we investigate a simple class of criteria that can guarantee satisfiability of a given $k$-CNF formula. We consider threshold criteria, i.e., for several quantities connected to a CNF (like the number of clauses, variables or variable intersections) we determine a maximum magnitude leading to satisfiable formulas. We would like to determine the exact threshold of such quantities, in the sense that there exist unsatisfiable formulas for which this quantity is greater than the threshold. A tightly determined threshold can be used as a simple satisfiability test: given a formula $F$, determine or count the specific quantities in $F$ and declare $F$ satisfiable if one of these quantities is below the threshold. Observe that such thresholds help in deciding satisfiability only if the considered quantity are below the threshold. The problem of deciding satisfiability when all these quantities are above the threshold is still a hard problem.\\

One such threshold that we consider is the \emph{number of clauses} $m$. We denote this threshold by $\mu_m(k)$ and it denotes the smallest number of clauses in an unsatisfiable formula. The trivial lower bound of $\mu_m(k)\geq 2^k$ is easily seen: each formula that consists of less than $2^k$ clauses is satisfiable since each clause eliminates only one out of the $2^k$ possible satisfying assignments. On the other hand there is an unsatisfiable $k$-CNF formula with $2^k$ clauses, namely the formula consisting of all possible $2^k$ clauses (all positive/negative literal combinations) on $k$ variables. Hence, $\mu_m(k)=\Theta(2^k)$.\\ 

Yet another prominent threshold is the \emph{maximum clause degree} $\Delta$ of a $k$-CNF formula, i.e. the maximum number of clauses that share at least one variable with a fixed clause. The complete formula on $k$ variables once again has maximum degree $2^k$ and gives an easy upper bound for this threshold. On the other hand an application of the powerful Lov\'asz Local Lemma \cite{ErdoesLovasz} shows that every formula with $\Delta <  2^k/e$ is satisfiable leading to the conclusion that $\mu_{\Delta} = \Theta(2^k)$.\\

In this paper we focus on satisfiability-threshold for a special class of formulas which guarantee that two clauses intersect only in a bounded number (henceforth we denote this by $\alpha$) of variables. These formulas are a natural extension of linear CNF formulas, i.e., formulas with $\alpha = 1$, which have been introduced in \cite{porschen2009linear}. The naming and concept of linear CNF formula comes from hypergraphs with bounded intersections as studied for example in \cite{ErdoesLovasz}. Intuitively, the restriction to bounded intersection makes it harder to build conflicting clauses which lead to unsatisfiability. And indeed it was the original goal of the authors to prove a higher satisfiability-threshold for $\Delta$ in linear $k$-CNF using stronger versions of the LLL, e.g., the soft-core LLL version of \cite{scott2005repulsive}. While it turned out that the satisfiability threshold for $\Delta$ remains $\tilde{\Theta}(2^k)$ even for linear CNFs we got interesting dependencies on $\alpha$ in the thresholds for other quantities, namely the number of variables, the number of clauses and the number of clause intersection pairs.

\section{Related work}

This paper builds highly on the techniques developed by Erd\H{o}s and Lov\'asz in the classical paper ``Problems and results on 3-chromatic hypergraphs and some related questions'' \cite{ErdoesLovasz}. Our proofs are built on the powerful Lov\'asz Local Lemma and also make use of and extend the shrinking operation (see Section \ref{sec:shrinking-hypergraphs}) that was used in \cite{ErdoesLovasz} to construct interesting linear hypergraphs. Independently but roughly a year before the authors conducted this research the paper \cite{scheder08almostdisjoint} by Dominik Scheder examined the satisfiability threshold for the number of clauses/constraints applying essentially the same techniques as here and in \cite{ErdoesLovasz}. While Scheder considers multi-value constraint satisfaction problems -- essentially a non-binary variant of CNF formula---he restricts himself to the threshold $\mu_m$. All results presented here directly extend to these multi-value CSPs too and to our knowledge this paper is the first to states the thresholds for the number of clause intersection pairs, variables and the max degree explicitly. More complicated algebraic constructions based on ideas of Kuzjurin \cite{Kuzjurin} and Kostochka and R\"{o}dl \cite{roedl} work for the restricted case $\alpha=1$ and can be found in Lemma 2.2. of \cite{Scheder10} without explicit statement of thresholds. Most notably, we use the $\alpha$-shrinking procedure not just in the lower bound but apply it to a maximal $(k+\alpha)$-uniform $\alpha$-intersecting formula in our upper bound construction. This is the key to obtaining bounds on the number of clause intersections and gives an unsatisfiable $\alpha$-intersecting formula that is extremal (up to $\log$-factors) in all considered quantities. \\

Another very interesting related work by Scheder and Zumstein is the paper ``How many conflicts does it need to be Unsatisfiable''~\cite{scheder08conflicts} in which upper and lower bounds on the threshold for conflicts are given. The notion of a conflict is closely related to clause intersections. Instead of counting the pairs of clauses that share a variable the number of conflict only counts clause pairs in which at least one variable is shared in an opposite literal. The reason why conflicts are interesting is because the lopsided version of the Lov\'asz Local Lemma \cite{erds1991lopsided} can be applied to $k$-CNF formulas in which each clause is involved in at most $2^k/e$ conflicts and thus guarantees their satisfiability. In contrast to the nearly tight threshold $\mu_i(k,\alpha) = \tilde{\Theta}(2^{k(2+1/\alpha)})$ for clause intersections in $\alpha$-intersecting formula established here, the conflict threshold is much harder to determine: the best known result for $\alpha=k$  is $\omega(2.69^k) \leq \mu_c(k,k) \leq O(4^k \frac{\log^2 k}{k})$ \cite{scheder08conflicts}.

\section{Preliminaries}

A hypergraph is $\mathbf k${\bfseries-uniform} if all edges contain exactly $k$ vertices. Two edges are called {\bfseries intersecting} if they share at least one vertex and a hypergraph is called $\mathbf \alpha${\bfseries-intersecting} if any two intersecting edges share at most $\alpha$ vertices. A $1$-intersecting hypergraph is called {\bfseries linear}. The {\bfseries edge intersection pairs} of a hypergraph are all pairs of edges that are intersecting. The {\bfseries degree of a vertex} is the number of edges it appears in and the {\bfseries degree of an edge} is the number of edges it intersects with.\\

Every $k$-CNF formula $F$ {\bf induces} a $k$-uniform (multi)-hypergraph $G_F=(V,E)$ where $V$ is the set of variables and the edge (multi)-set $E$ contains an hyperedge over vertices $\{v_1,\cdots,v_k\}$ if and only if there exists a clause consisting of the corresponding variables. This gives a one-to-one mapping between clauses and edges in the induced hypergrah and we adopt all previously introduced hypergraph terminology for $k$-CNF formula accordingly, e.g., we define clause intersection pairs as all pairs of clauses that intersect in at least one variable.\\

Throughout this paper we are interested in satisfiability thresholds for $\alpha$-intersecting $k$-CNF formula. We consider the following quantities: number of clauses $m$, number of variables $n$, maximum degree $\Delta$ and number of clause intersection pairs $i$. Denote the thresholds for a quantity $q$ with $\mu_q(\alpha,k)$. A {\bfseries satisfiability threshold} $\mu_q(\alpha,k)$ is the smallest number such that there exists an unsatisfiable $\alpha$-intersecting $k$-CNF with $q=\mu_m(\alpha, k)$. Phrased differently it is the largest number such that every $\alpha$-intersecting $k$-CNF formula with $q < \mu_q(\alpha, k)$ is satisfiable.\\

Our lower bounds to the thresholds are based on a classical application of the Lov\'asz Local Lemma \cite{ErdoesLovasz} and its more recent constructive algorithmic versions that give randomized \cite{moser08} and deterministic \cite{MT-JACM,llldeterministic} algorithms:

\begin{theorem}\label{thm:lll}
Every $k$-CNF with maximum clause degree $\Delta$ at most $\frac{2^k}{e}$ is satisfiable and there is an efficient algorithm to find such an assignment.
\end{theorem}

\section{Results}

We present lower bounds (Theorem \ref{thm:lowerbound}) and nearly matching constructive upper bounds (Theorem \ref{thm:upperbound}) that determine all thresholds $\mu_i,\mu_m,\mu_n,\mu_\Delta$ up to $\log$-factors (Theorem \ref{thm:thresholds}). Our lower bound in Theorem \ref{thm:lowerbound} consists of an algorithm based on Theorem \ref{thm:lll} that efficiently finds a satisfying assignment for any $\alpha$-intersecting $k$-CNF formula with few clause intersection pairs, variables or clauses. The upper bound in Theorem \ref{thm:upperbound} proves the existence of unsatisfiable formulas which have only slightly more clause intersections, variables and clauses. Note that while our proof of Theorem \ref{thm:lowerbound} is algorithmic, one needs an efficient implementation of Lemma \ref{lemma:construction-unsat} to make Theorem \ref{thm:upperbound} constructive(see also \cite{Scheder10}). We suspect that some of the bounds below can be improved by $O(k)$-factors but since all bounds are exponential in $k$ we did not optimize for these polylogarithmic factors.

\begin{theorem}\label{thm:lowerbound}
Every $\alpha$-intersecting $k$-CNF with less than 
$$L_i = \frac{1}{2\alpha} \left(\frac{2^{(k-\alpha)}}{ek}-1\right)^{(2+1/\alpha)} \text{   clause intersections} $$ 
or
$$L_n =  \left(\frac{2^{(k-\alpha)}}{ek}\right)^{1/\alpha} \text{   variables}$$ 
or 
$$L_m = \frac{1}{k} \left(\frac{2^{(k-\alpha)}}{ek}\right)^{1+1/\alpha} \text{    clauses}$$ 
is satisfiable and a satisfying assignment can be found efficiently.
\end{theorem}

\medskip

\begin{theorem}\label{thm:upperbound}
For any $k$ and $\alpha < k$ there is an unsatisfiable $\alpha$-intersecting $k$-CNF with at most
$$U_i = \alpha^2 2^{(k+\alpha)(2+1/\alpha)} k^{(5+2/\alpha)}  \text{   clause intersections}$$
and 
$$U_n = 2\alpha 2^{k/\alpha} k^{2(1+1/\alpha)}   \text{   variables}$$
and
$$U_m = \alpha 2^{(k+\alpha)(1+1/\alpha)}{k^{2(1+1/\alpha)}} \text{    clauses}$$
and 
$$U_{\Delta} = \alpha 2^{(k+\alpha)}k^2 \text{   maximum degree}.$$
\end{theorem}

\medskip

In the following $\tilde{\Theta}(x)$ means $\Theta(x (\log x)^c)$ for some absolute 
(positive or negative) constant $c$. Combining the above two theorems yields good estimates for the thresholds: 

\begin{corollary}\label{thm:thresholds}
The thresholds for satisfiability are:
\begin{itemize}
	\item number of clause intersections: $\mu_{i} = \tilde{\Theta}(2^{k(2+1/\alpha)})$
	\item number of variables: $\mu_n = \tilde{\Theta}(2^{k/\alpha})$
	\item number of clauses: $\mu_m = \tilde{\Theta}(2^{k(1+\frac{1}{\alpha})})$
	\item maximum degree:	$\mu_{\Delta} = \tilde{\Theta}(2^{k})$
\end{itemize}
\end{corollary}

\section{Shrinking and Maximal $\alpha$-intersecting Hypergraphs}\label{sec:shrinking-hypergraphs}
This section contains useful lemmas about hypergraphs needed to prove the main theorems. One operation that will be particularly helpful for both the lower and the upper bound is the $\mathbf \beta${\bfseries -shrinking} operation. The shrinking operation creates a $k$-uniform hypergraph $H'$ from a $(k+\beta)$-uniform hypergraph $H$ by deleting the $\beta$ vertices of maximum degree from each edge breaking ties arbitrarily. Shrinking is similarly defined for $(k+\beta)$-CNF formulas where the variables with highest degree are deleted from each clause. The next lemma shows that a high degree vertex can survive the $\beta$-shrinking procedure to remain a high degree vertex only if many such high degree vertices are present in the original hypergraph.\\
 
\begin{lemma}\label{lemma:shrinking}
Let $H$ be a $(k+\alpha)$-uniform $\alpha$-intersecting hypergraph and $H'$ be the result of $\alpha$-shrinking $H$. If $H'$ has a vertex of degree $d$, then $H$ has more than $d^{1/\alpha}$ vertices of degree at least $d$. 
\end{lemma}

\begin{proof}
Let $v$ be the vertex in $H'$ of degree $d$. Since $H'$ was created by shrinking $H$ there are at least $d$ edges in $H$ in which $v$ is present but did not get deleted. We call the set of those edges $\C$; then we know that $|\C|\geq d$. From each edge $e \in \C$, exactly $\alpha$ vertices got deleted all of which are of degree of at least $d$. We claim that the mapping that maps each $e \in \C$ to this $\alpha$-sized set of deleted vertices is injective:\\
 
Suppose two edges $e_1, e_2\in \C$ get mapped to the same $\alpha$-sized set of vertices. Then, the edges $e_1$ and $e_2$ intersect in these $\alpha$ vertices; furthermore they also intersect in the vertex $v$ and thus intersect in $\alpha+1$ vertices. This is a contradiction to the $\alpha$-intersecting property of $H$.\\

Injectivity gives us that there are $|\C| \geq d$ different $\alpha$-sized subsets of vertices which got deleted instead of $v$ while shrinking. All vertices in those subsets must have degree at least $d$ by definition of the shrinking operation. Furthermore if $N$ is the number of distinct vertices in those subsets then we have $d \leq \binom{N}{\alpha} < N^{\alpha}$. Therefore there are at least $N > d^{1/\alpha}$ vertices with degree at least $d$ in $H$. 
\end{proof}

\medskip

The next lemma proves that any maximal $\alpha$-intersecting hypergraph on $n$ vertices must have a large number of edges. It uses a bound on the Tur{\'a}n number that is due to de Caen \cite{de1983extension}. The Tur{\'a}n number $T(n,k,r)$ for $r$-uniform hypergraphs with $n$ vertices is the smallest number of edges possible such that every set of $k$ vertices contains at least one edge. This number was determined for graphs by Tur{\'a}n \cite{turan1941extremal} and extended to hypergraphs by himself in the report ''Research Problems''\cite{turan1961research}.\\

\begin{lemma}\label{lem:maximal-hypergraphs}
Every maximal $\alpha$-intersecting hypergraph $H$ on $n$ vertices has at least $m \geq \frac{\binom{n}{\alpha+1}}{\binom{k}{\alpha+1}^2}$ edges. 
\end{lemma}
\begin{proof}
Let $H$ be a maximal $\alpha$-intersecting hypergraph on $m$ edges. Since $H$ is $\alpha$-intersecting each of the $\binom{n}{\alpha+1}$ subsets of vertices of size $\alpha+1$ is covered by at most one distinct hyperedge of $H$. Also, $H$ covers exactly $m \binom{k}{\alpha+1}$ distinct subsets of size $\alpha+1$ in $H$. If $m{\binom{k}{\alpha+1}} < {T(n,k,\alpha+1)}$ the $\alpha+1$-uniform hypergraph consisting of all covered $\alpha+1$-size subsets has less than $T(n,k,\alpha+1)$ edges and therefore $\exists$ a $k$-subset $K$ that does not contain any covered edge. This $k$-subset can be added as an edges into $H$ while preserving it to be $\alpha$-intersecting. Indeed, if some edge $e$ intersects $K$ in at least $\alpha+1$ vertices, then the corresponding set of vertices is covered contradicting the choice of $K$. Thus if $m<\frac{T(n,k,\alpha+1)}{\binom{k}{\alpha+1}}$ then $H$ is not maximal $\alpha$-intersecting. To finish we use a lower bound of de Caen \cite{de1983extension} on the Tur{\'a}n number: $T(n,k,\alpha+1) \geq \frac{n-k+1}{n-\alpha}\binom{n}{\alpha+1}/\binom{k-1}{\alpha}$; plugging this in gives the desired result.
\end{proof}

We remark that the same result also appears in Scheder 
\cite{scheder08almostdisjoint} with somewhat simpler and
self-contained proof.

\section{A Constructive Lower Bound}

This section gives the proof for the lower bound in Theorem \ref{thm:lowerbound}:\\

\begin{proof} (of Theorem \ref{thm:lowerbound})\\
We prove that every $\alpha$-intersecting $k$-CNF $F$ is either satisfiable by Theorem \ref{thm:lll} after $\alpha$-shrinking it or it must have large clause intersection pairs, variables, clauses and a high maximum degree contradicting the hypothesis about the formula $F$.\\

Let $F'$ be the resulting $(k-\alpha)$-CNF we get from $\alpha$-shrinking $F$. If all variables in $F'$ have degree less than $d = 2^{(k-\alpha)}/ek$ then the Lov\'asz Local Lemma guarantees that $F'$ is satisfiable and Theorem \ref{thm:lll} states that a satisfying assignment can be efficiently found. Note that a satisfying assignment for $F'$ is also a satisfying assignment for $F$.\\

In the other case, suppose $F'$ has at least one variable of degree $d$. Then, Lemma \ref{lemma:shrinking} shows that $F$ must have at least $d^{1/\alpha}$ variables of degree at least $d$.\\

To count the number of clause intersection pairs in $F$, we count the intersections of clauses containing one of the $d^{1/\alpha}$ high degree variables. For each such variable the clauses containing it induce a clique with $(d-1)^2/2 $ intersections. Taking the disjoint union of these intersections we get at least $(d-1)^{2+1/\alpha}/2$ intersections but overcount each intersection up to $\alpha$-times since two clauses can intersect in up to $\alpha$ variables. Therefore $F$ has at least $\frac{1}{2\alpha} (d-1)^{2+1/\alpha}$ intersections.\\

To count the number of clauses in $F$ we look at the union of the clauses containing one of the $d^{1/\alpha}$ variables. There are at least $d^{1+1/\alpha}$ clauses in the non disjoint union and each clause can get added because of each of its $k$ variables at most once. Thus $F$ has at least $d^{1+1/\alpha}/k$ clauses.\\

Finally it is clear that $F$ has at least $d^{1/\alpha}$ variables. 
\end{proof}

\section{Upper bounds for the thresholds}

This section gives the proof for the upper bounds in Theorem \ref{thm:lowerbound}.\\

Before we prove the theorem itself, the following lemma gives a general way to transform a sufficiently dense $k$-uniform hypergraph into an unsatisfiable $k$-CNF formula by iteratively taking a hyperedge and greedily choosing positive or negative literals for the variables:\\

\begin{lemma}\label{lemma:construction-unsat}
If there is a $k$-uniform hypergraph $H$ on $n$ vertices and at least $m = n 2^k$ edges than there exists an unsatisfiable $k$-CNF $F$ inducing $H$. 
\end{lemma}
\begin{proof}
Denote the vertices in $H$ by $v_1,\ldots,v_n$ and associate with them, the variables $x_1,\cdots,x_n$ that will occur in $F$.

We will denote $\A\in\{0,1\}^n$ to be an assignment if we pick an assignment of values to variables $x$ by setting $x_i=\A_i$.
We say that a clause \emph{covers} an assignment $\A\in\{0,1\}^n$ if it is not satisfied by the assignment. We will iteratively create a clause for every edge in $H$ greedily covering the maximum number of yet uncovered assignments. We have to show that in the end all $2^n$ assignments are covered. Consequently, the conjunction of the created clauses forms an unsatisfiable $k$-CNF.\\

We pick edges $e$ from $H$ in an arbitrary order. We want to create a clause for $e$ on the $k$ variables associated with the $k$ vertices in $e$. For each variable we have the choice to pick the positive or the negative literal. These are $2^k$ different choices and the assignments covered by two different choices are disjoint. Since every assignment can be covered in this way the assignments get partitioned into $2^k$ parts. Simple averaging then guarantees that there exists a choice covering at least $1/2^k$ fraction of the assignments not covered so far. After $m$ iterations of greedily creating clauses covering the maximal number of uncovered assignments is at most $2^n \left(1 - 1/2^k\right)^m = \left(\frac{2}{(1 - 2^{-k})^{2^{-k}}}\right)^n < 1$. With all assignments covered the created formula $F$ is unsatisfiable and by construction also induces $H$ as required.
\end{proof}

\medskip

The above lemma shifts the focus towards finding a suitable dense $k$-uniform hypergraph in order to find an unsatisfying $k$-CNF. The following proof of Theorem \ref{thm:lowerbound} shows that $\alpha$-shrinking a maximal $\alpha$-intersecting $(k+\alpha)$-uniform hypergraph results in hypergraphs with nice additional extremal properties. Furthermore choosing a large number of vertices results in hypergraphs that obey the bound in Lemma \ref{lemma:construction-unsat} and can thus be transformed into the desired unsatisfiable $k$-CNF.

\begin{proof} (of Theorem \ref{thm:upperbound})\\
We create the formula by applying Lemma \ref{lemma:construction-unsat} to an $\alpha$-intersecting hypergraph. We obtain this hypergraph by  $\alpha$-shrinking a maximal $\alpha$-intersecting $(k+\alpha)$-uniform hypergraph. Observe that it makes the resulting hypergraph $k$-uniform.\\

We choose $n = \alpha \left(2^{k+\alpha}k^{2(\alpha+1)}\right)^{1/\alpha}$ and build a $\alpha$-intersecting $(k+\alpha)$-uniform hypergraph on $n$ vertices. The choice of $n$ is such that Lemma \ref{lem:maximal-hypergraphs} guarantees that we can find a $k+\alpha$-uniform hypergraph $H$ with $$m = \frac{n^{\alpha+1}}{k^{2(\alpha+1)}\alpha^\alpha} = \alpha 2^{(k+\alpha)(1+1/\alpha)}{k^{2(1+1/\alpha)}} = 2^{k+\alpha}n$$ edges. This is sufficiently large number of edges to construct an unsatisfiable formula $F$ for hypergraph $H$ using Lemma \ref{lemma:construction-unsat}. Having constructed $H$, we $\alpha$-shrink it to obtain hypergraph $H'$ and its correspoding formula $F'$.  Note that $F'$ is unsatisfiable because $F$ is unsatisfiable.  The significant advantage about $H'$ obtained this way is that it has guarantees on the maximum degree and on the number of clause intersections. More precisely, we claim that $H'$ has maximum degree less than $(mk)^{1/(1+1/\alpha)}$. Suppose that after the shrinking there is a vertex of degree $d > (m(k+ \alpha))^{1/(1+1/\alpha)}$. Lemma \ref{lemma:shrinking} shows that in this case $H$ contains at least $d^{1/\alpha}$ vertices of degree larger than $d$. The disjoint union of the edges containing those vertices has size at least $d^{1+1/\alpha}$ and each edge gets counted at most $(k+\alpha)$ times this way. Therefore $H$ would have at least $d^{1+1/\alpha}/(k+\alpha) > m$ edges --- a contradiction.\\

Lemma \ref{lemma:construction-unsat} transforms the hypergraph $H$ into an unsatisfiable $k$-CNF formula $F$. This formula has $n$ variables and $m$ edges since shrinking preserves these quantities. Furthermore, the maximum degree $\Delta$ of $F$ is at most $(mk)^{1/(1+1/\alpha)}$ which also implies that the number of clause intersections is at most $m\Delta$.
\end{proof}

\medskip

{\bfseries Acknowledgments}\\
The research for this paper was done in the summer of 2009 while the authors where at MSR India. We thank Aravind Srinivasan for pointing out the questions about satisfiability thresholds for almost disjoint $k$-CNF formula which lead to this paper. We also want to thank Dominik Scheder.

\bibliographystyle{abbrv}    
\bibliography{alphaCNF}

\begin{thebibliography}{10}

\bibitem{llldeterministic}
K.~Chandrasekaran, N.~Goyal, and B.~Haeupler.
\newblock {Deterministic Algorithms for the Lov{\'a}sz Local Lemma}.
\newblock In {\em Proceedings of ACM-SIAM Symposium on Discrete Algorithms
  (SODA)}, 2010.

\bibitem{de1983extension}
D.~De~Caen.
\newblock {Extension of a theorem of Moon and Moser on complete subgraphs}.
\newblock {\em Ars Combinatoria 16}, pages 5--10, 1983.

\bibitem{ErdoesLovasz}
P.~Erd\H{o}s and L.~Lov{\'a}sz.
\newblock {Problems and results on 3-chromatic hypergraphs and some related
  questions}.
\newblock {\em Infinite and finite sets}, 2:609--627, 1975.

\bibitem{erds1991lopsided}
P.~Erd\H{o}s and J.~Spencer.
\newblock {Lopsided Lov{\'a}sz local lemma and latin transversals}.
\newblock {\em Discrete Applied Mathematics}, 30(2-3):151--154, 1991.

\bibitem{roedl}
A.~V. Kostochka and V.~R\"{o}dl.
\newblock Constructions of sparse uniform hypergraphs with high chromatic
  number.
\newblock {\em Random Struct. Algorithms}, 36(1):46--56, 2010.

\bibitem{Kuzjurin}
N.~N. Kuzjurin.
\newblock On the difference between asymptotically good packings and coverings.
\newblock {\em Eur. J. Comb.}, 16(1):35--40, 1995.

\bibitem{moser08}
R.~A. Moser.
\newblock A constructive proof of the lov\'{a}sz local lemma.
\newblock In {\em STOC '09: Proceedings of the 41st annual ACM symposium on
  Theory of computing}, pages 343--350, New York, NY, USA, 2009. ACM.

\bibitem{MT-JACM}
R.~A. Moser and G.~Tardos.
\newblock A constructive proof of the general lov\'{a}sz local lemma.
\newblock {\em Journal of the ACM}, 57(2):1--15, 2010.

\bibitem{porschen2009linear}
S.~Porschen, E.~Speckenmeyer, and X.~Zhao.
\newblock {Linear CNF formulas and satisfiability}.
\newblock {\em Discrete Applied Mathematics}, 157(5):1046--1068, 2009.

\bibitem{scheder08almostdisjoint}
D.~Scheder.
\newblock Satisfiability of almost disjoint cnf formulas.
\newblock {\em CoRR}, abs/0807.1282, 2008.

\bibitem{Scheder10}
D.~Scheder.
\newblock {Unsatisfiable Linear CNF Formulas Are Large and Complex}.
\newblock {\em STACS 2010, 27th International Symposium on Theoretical Aspects
  of Computer Science}, pages 621--632, 2010.

\bibitem{scheder08conflicts}
D.~Scheder and P.~Zumstein.
\newblock {How many Conflicts does it need to be Unsatisfiable?}
\newblock {\em Theory and Applications of Satisfiability Testing--SAT 2008},
  pages 246--256, 2008.

\bibitem{scott2005repulsive}
A.~Scott and A.~Sokal.
\newblock {The repulsive lattice gas, the independent-set polynomial, and the
  Lov{\'a}sz local lemma}.
\newblock {\em Journal of Statistical Physics}, 118(5):1151--1261, 2005.

\bibitem{turan1941extremal}
P.~Tur{\'a}n.
\newblock {On an extremal problem in graph theory}.
\newblock {\em Mat. Fiz. Lapok}, 48:436--452, 1941.

\bibitem{turan1961research}
P.~Tur{\'a}n.
\newblock {\em {Research problems}}.
\newblock Akad. Kiad{\'o}, 1961.

\end{thebibliography}

\end{document}